\documentclass{eptcs}
\usepackage{amsmath}
\usepackage{amssymb}
\usepackage{amsthm}
\usepackage{mathtools}
\usepackage{xspace}
\newcommand{\Iter}{\ensuremath{\mathcal{I}}\xspace}
\newcommand{\Jter}{\ensuremath{\mathcal{I}'}\xspace}
\newcommand{\Rec}{\ensuremath{\mathcal{R}}\xspace}
\newcommand{\mpt}{\ensuremath{\mathsf{MPT}}\xspace}
\newcommand{\spt}{\ensuremath{\mathsf{SPT}}\xspace}
\newcommand{\opt}{\ensuremath{\mathsf{OPT}}\xspace}
\newcommand{\otm}{\ensuremath{\mathsf{OTM}}\xspace}
\newcommand{\otms}{\ensuremath{\mathsf{OTM}}s\xspace}
\newcommand{\p}{\ensuremath{\mathsf{P}}\xspace}
\newcommand{\str}{\mathbf}
\newcommand{\lrn}{\ensuremath{\mathsf{lrn}}\xspace}
\newcommand{\redP}{\preceq_{\lambda\p}}
\newcommand{\equivP}{\equiv_{\lambda\p}} 
\renewcommand{\max}{\ensuremath{\mathsf{max}}\xspace}
\newcommand{\cut}{\mathbin{\scriptstyle\dot{\smash{\textstyle-}}}}
\newcommand\OMIT[1]{}
\newcommand{\sdzero}{\texttt0}
\newcommand{\sdone}{\texttt1}
\newcommand{\demph}{\emph}
\newcommand{\lmin}{\mathsf{lmin}}
\newcommand{\authorrunning}{B.M.~Kapron and F.~Steinberg}
\newcommand{\titlerunning}{Type-two Iteration with Bounded Query Revision}
\newtheorem{theorem}{Theorem}[section]
\newtheorem{lemma}[theorem]{Lemma}
\newtheorem{cor}[theorem]{Corollary}
\theoremstyle{definition}
\newtheorem{definition}[theorem]{Definition}
 
\title{Type-two Iteration with Bounded Query Revision}
\author{Bruce M. Kapron\thanks{Research supported in part by an NSERC Discovery Grant}\\
University of Victoria\\
Victoria, BC, Canada\\
\texttt{bmkapron@uvic.ca}
\and
Florian Steinberg\\
INRIA\\
Saclay, \^{I}le-de-France\\
\texttt{florian.steinberg@inria.fr}
}

\begin{document}

\maketitle

\begin{abstract}
Motivated by recent results of Kapron and Steinberg (LICS 2018) we introduce new forms of iteration on length in the setting of applied lambda-calculi for higher-type poly-time computability.
In particular, in a type-two setting, we consider functionals which capture iteration on input length which bound interaction with the type-one input parameter, by restricting to a constant either the number of times the function parameter may return a value of increasing size, or the number of times the function parameter may be applied to an argument of increasing size.
We prove that for any constant bound, the iterators obtained are equivalent, with respect to lambda-definability over type-one poly-time functions, to the recursor of Cook and Urquhart which captures Cobham's notion of limited recursion on notation in this setting.
\end{abstract}

\section{Introduction}

Recursion on notation is a fundamental tool for syntactic characterizations of feasible computation, in particular capturing the notion of bounding the number of steps of a computation in terms of input size.
However, as a constraint, it is too weak on its own to capture feasibility as characterized by polynomial time computability.
A well-known example is the following: consider a function $\varphi$, mapping binary strings to binary strings, which for any input string $\str a$ returns the string concatenated with itself: $\varphi(\str a)=\str a\str a$.
The function $\varphi$ should clearly be accepted as feasible, but recursion on notation allows the definition of a new function which, on input $\str b$ of length $n$ returns $\varphi^{n}(\sdzero)$, which is a string of length $2^n$.
To capture feasibility through a recursion scheme, further restrictions are required to prevent this kind of exponential blowup.
Indeed, Cobham, in perhaps one of the earliest works mentioning polynomial-time computability, gives a characterization using a scheme of \emph{limited} recursion on notation \cite{Cobham}.
Here, definition of a new function through recursion on functions already known to be from the class is allowed only in case the length of the resulting function may be \emph{a priori} bounded by the length of a function already known to be in the class.
Cobham's approach is a canonical example of \emph{explicit} bounding.
It is also possible to formulate forms of limited recursion with \emph{implicit} bounding and recover the same class of polynomial time functions \cite{Leivant, BC}.

It is possible to consider feasibility in the \emph{type-two} setting, which allows computation with respect to an arbitrary function oracle.
The original definition of type-two polynomial time was given by Mehlhorn using a straightforward generalization of Cobham's scheme \cite{Mehlhorn}.
Just like the polynomial time functions, this class of functionals allows for a number of different characterizations and is accepted as capturing feasibility at type level two appropriately:
Cook and Urquhart gave a formulation of Mehlhorn's class, and in fact generalized it to all finite types by use of an applied typed lambda calculus with constant symbols for a collection of basic type-one poly-time functions, as well as a \emph{recursor} $\Rec$, which captures Mehlhorn's scheme as a type-two functional \cite{CU}.
Kapron and Cook showed that Mehlhorn's class may be characterized in terms of oracle Turing machines (\otms) whose run-time is bounded by a \emph{second-order polynomial} \cite{KC}.
Both of these characterizations have lead to a multitude of applications and further characterizations.

The content of this paper is inspired by a recent description of Mehlhorn's class given by Kapron and Steinberg \cite{KS}.
For this it is instructive to think of an analogue of unrestricted recursion on notation in the \otm setting.
Informally, this corresponds to Cook's notion of \emph{oracle polynomial time} (\opt) \cite{Cook}, which bounds the running time of \otms by a polynomial in the size of the input and the largest answer returned by any call to the oracle.
Here, a higher time consumption can be justified by an increasing chain of oracle return values and in particular it is possible to recover the example above within \opt.
To force feasibility, Kapron and Steinberg use restrictions of \opt based on \emph{query-size revisions}.
They considered two forms of revision: a \emph{length revision} occurs when a query to the oracle returns an answer with size larger than the size of the input or the answer to any previous query, a \emph{lookahead revision} occurs when the size of a query provided to the oracle is larger than the size of any previous such query.
{\em Strong polynomial time} (\spt) allows only a constant number of length revisions, while {\em moderate polynomial time} (\mpt) allows only a constant number of lookahead revisions.
Kapron and Steinberg prove that both of the classes \spt and \mpt give proper subsets of Mehlhorn's class even when restricted to the functionals of the type that they are meant to capture, but that Mehlhorn's class can be recovered from each of the classes by closing under $\lambda$-abstraction and application.
It should be noted that length revisions and \spt make an earlier apparent in a somewhat different setting in work of Kawamura and Steinberg \cite{DBLP:conf/rta/KawamuraS17}.

The outline of this paper is as follows:
In the first section we describe the setting.
Namely we work in a simply typed lambda-calculus with constant symbols for all type-1 polynomial-time computable functions.
This is identical to the setting Cook and Urquhart chose for their characterization of higher-order polynomial time through the recursor $\Rec$ and means that we reason about higher-order complexity modulo the availability of the full strength of a first-order bounded recursion scheme.
The paper starts from the observation that the bounded recursor $\Rec$ is meant to model Mehlhorn's scheme, which is strictly more expressive than the first order scheme that is already available through the constants.
Clearly $\Rec$ adds something, as the class of functionals expressible without its presence has been classified by Seth and is considerably restricted in its access of the oracle \cite{MR1238294}.
Thus, one may ask for functionals that are less expressive and still generate the same class given the context.
Section \ref{Cook Urquhart} weakens $\Rec$ in two steps by first simplifying the way in which the bounding is done and afterwards by restricting the data that is available to the step-function.
This leaves us with a functional $\Iter$ that no longer captures bounded recursion but is better understood as doing bounded iteration.

Section \ref{length revision} starts involving the ideas of length revisions:
Inspired by the definitions of the classes \spt and \mpt we change the way in which iteration is bounded.
The new conditions intuitively provide more freedom than the direct bounding the iterator $\Iter$ uses and do so in a way that is somewhat orthogonal to how Cook and Urquhart's original recursor $\Rec$ did more complicated bounding.
We are lead to consider a family of operators $\Iter_k$ where the condition that is imposed becomes less restrictive as $k$ grows.
Over the chosen background theory, all of the operators $\Iter_k$ as well as $\Rec$ and $\Iter$ are of equal expressive power.
However, the parameter $k$ is tightly connected to runtime-bounds for $\Iter_k$ in the \otm setting, and the use of higher values should allow expressing some functionals that feature more complicated interaction with the oracle more concisely.
The proof that all considered operators are equivalent additionally covers a similarly defined family of iterators based on the idea of lookahead revisions that is introduced in Section~\ref{lookahead revision}.
The final section specifies an efficient generation scheme for the values of the new iterators.

Kapron and Steinberg define the classes \spt and \mpt using the \otm framework which is bound to a specific machine model.
This paper transfers the notions of length and lookahead revisions to the machine independent setting of iteration schemes, where the number of iterations is determined by the length of a specified input parameter (which is a string over some finite alphabet).
Our proofs introduce some interesting and useful idioms for programming in this setting.

\subsection{Preliminaries}
Let $\Sigma$ denote a finite alphabet that contains symbols $\sdzero$ and $\sdone$, and $\Sigma^*$ the set of finite strings over $\Sigma$.
The empty string is denoted $\epsilon$, and arbitrary elements of $\Sigma^*$ are denoted $\str a, \str b, \str c, \dots$.
We attempt to bind names of string variables to their meanings as far as possible:
$\str a$ is associated with initial values, $\str b$ with size-bounds $\str c$ with values that a recursion or iteration is carried out over and $\str t$ the previous values in a recursion or iteration.
For $\str a \in \Sigma^*$ let $|\str a|$ to denote the length of $\str a$ and $a_i$ its digits, i.e., $\str a=a_1\dots a_{|\str a|}$.
We write $\str b \subseteq \str a$ to indicate that $\str b$ is an initial segment of $\str a$.
We assume that we have symbols for all type-1 poly-time functions, for instance:
\begin{itemize}
\item {\em Truncation}:
  The 2-ary function sending $\str b$ and $\str c = c_1c_2\dots c_{|\str c|}$, to $\str c^{\leq |\str b|}:=c_1\dots c_{|\str b|}$, if $|\str b|\leq |\str c|$ and $\str c$ otherwise.
  Note that always $\str c ^{\leq |\str b|}\subseteq \str c$ and $\str c^{\leq |\varepsilon|}=\varepsilon$.
  For $\str c^{\leq |\str c|-1}$ we use the shorthand $\str c\gg1$.
\item {\em Tupling} and {\em projection} functions $\langle\cdot,\ldots,\cdot\rangle$ and $\pi_i$, such that tupling is monotone with respect to length in each argument.
  Namely if $|\str a_i|=|\str b_i|$ for all $i$ apart from $k$, then $|\str a_k| \leq |\str b_k|$ if and only if $|\langle \str a_1, \dots, \str a_n\rangle| \le |\langle \str b_1, \dots, \str b_n\rangle|$.
\item {\em Length minimum}: We adopt the convention used by Cook and Urquart, i.e.\
  \[ \lmin(\str c, \str b) := \begin{cases} \str c & \text{ if } |\str c| < |\str b| \\ \str b & \text{otherwise.} \end{cases} \]   
\end{itemize}
We also use definition by cases extensively, relying on the fact that there is a polynomial-time conditional and avoid over-use of $\lambda$-abstractions via explicit function definition.
Tupling functions that satisfy the demands above exist and are 1-1, but not bijective.
In spite of this we still write $\lambda \langle \str a_1, \dots, \str a_k\rangle.t[\str a_1,\dots,\str a_k]$ as short hand for  $\lambda \str b.t[\pi_1\str b, \ldots,\pi_k\str b]$.
This is all done for the sake of readability.

\subsection[Lambda-Definability]{$\lambda$-Definability}\label{ap: lambda}
The treatment of the typed $\lambda$-calculus here follows that used by Cook and Urquart for their characterization of Mehlhorn's class \cite{CU}.
The set of {\em types} is defined inductively as follows:
\begin{itemize}
\item
0 is a type
\item
$(\sigma \rightarrow \tau)$ is a type, if $\sigma$ and $\tau$ are types.
\end{itemize}

The set $Fn(\tau)$ of  \demph{functionals of type} $\tau$ is defined
by induction on $\tau$:
\begin{itemize}
\item $Fn(0) = \Sigma^*$
\item $Fn(\sigma \rightarrow \tau) = \{F | F: Fn(\sigma) \rightarrow Fn(\tau)\}$.
\end{itemize}
It is not hard to show that each type $\tau$ has a unique normal form
\[ \tau = \tau_1 \rightarrow \tau_2 \rightarrow \cdots \rightarrow \tau_k \rightarrow 0 \]
where the missing parentheses are put in with association to the right.
Hence a functional $F$ of type $\tau$ is considered in a natural
way as a function of variables $X_1, \ldots , X_k$, with $X_i$
ranging over $Fn(\tau_i)$, and returning a natural number value:
$$F(X_1)(X_2) \ldots (X_k) = F(X_1, \ldots , X_k).$$

The \demph{level} of a type is defined inductively: The
level of type 0 is 0, and the level of the type $\tau$ written in the above normal form
is 1 + the maximum of the levels of $\tau_1, \ldots, \tau_k$.
This paper is mostly only concerned with functionals of type level smaller or equal two.

Let $\mathbf{X}$ be a class of functionals.
The set of $\lambda$-\demph{terms} over $\mathbf{X}$, denoted $\lambda(\mathbf{X})$ is
defined as follows:
\begin{itemize}
\item
For each type $\sigma$ there are infinitely many
variables $X^\sigma, Y^\sigma, Z^\sigma, \ldots$ of type $\sigma$, and each
such variable is a term of type $\sigma$.
\item
For each functional $F$ (of type $\sigma$) in $\mathbf{X}$ there is a term $F^\sigma$ of type $\sigma$.
\item
If $T$ is a term of type $\tau$ and $X$ is a variable
of type $\sigma$, then $(\lambda X.T)$ is a term of type
$(\sigma \rightarrow \tau)$ (an abstraction).
\item
If $S$ is a term of type $(\sigma \rightarrow \tau)$ and $T$ is a term
of type $\sigma$, then $(ST)$ is a term of type $\tau$
(an application).
\end{itemize}

For readability, we write $S(T)$ for $(ST)$; we also write $S(T_1, \ldots , T_k)$ for $(\ldots ((ST_1)T_2) \ldots T_k)$, and $\lambda X_1 \ldots \lambda X_k.T$ for $(\lambda X_1.(\lambda X_2.( \ldots (\lambda X_k.T) \ldots )))$.

The set of free variables of a lambda term can be defined inductively and are those that are not bound by a lambda abstraction.
A term is called closed, if it has no free variables.
In a natural way each closed $\lambda$-term $T$ of type $\tau$ represents a functional in $Fn(\tau)$.
This correspondence is demonstrated in the standard way, by showing that a mapping of variables to functionals with corresponding type can be extended to a mapping of terms to functionals with corresponding type.

An \demph{assignment} is a mapping $\varphi$ taking variables to functionals with corresponding type.
Suppose $\varphi$ is an assignment and $T$ a $\lambda$-term over $\mathbf X$.
The value $\mathcal{V}_\varphi(T)$ of $T$ with respect to $\varphi$ is defined by induction on $T$ as follows.

When $T$ is a variable, $\mathcal{V}_\varphi(T)$ is $\varphi(T)$.
If $T = F^\sigma$ is a constant symbol for some $F\in \mathbf X$, then $\mathcal{V}_\varphi(T) = F$.

Suppose that $\tau = \tau_1 \rightarrow \ldots \rightarrow \tau_k \rightarrow 0$.
When $T$ has the form $\lambda X^\sigma. S^\tau$, $F$ is a type $\sigma$ functional and $F_i$ are type $\tau_i$ functionals, then
\[ \mathcal{V}_\varphi(T)(F,F_1, \ldots, F_k)
:= \mathcal{V}_{\varphi'}(S)(F_1, \ldots, F_k), \]
where $\varphi'(X^\sigma) = F$, but $\varphi'$ is otherwise identical to
$\varphi$.
When $T$ has the form $S^{\sigma \rightarrow \tau}R^\sigma$,
\[ \mathcal{V}_\varphi(T)(F_1, \ldots, F_k) =
\mathcal{V}_\varphi(S)(\mathcal{V}_\varphi(R)(F_1, \ldots, F_k)).\hspace{.2in}\Box \]

It is not hard to show that if $T,S$ are terms such that $T$ is a $\beta$
or $\eta$ redex and $S$ is its contractum, then for all $\varphi$,
$\mathcal{V}_\varphi(T) = \mathcal{V}_\varphi(S)$.
A functional $F$ is \demph{represented} by a term $T$ relative to an assignment $\varphi$ if $F = \mathcal{V}_\varphi(T)$.

Our goal in this paper is to prove the equivalence, with respect to $\lambda$-representability in the presence of poly-time type-1 functions, of type-2 functionals capturing different forms of recursion on notation.
To this end we have the following definitions. 

\begin{definition}
  Let \p be the class of (type-1) poly-time functions, and $F,G$ be functionals.
  We say that $F$ is \p-\emph{reducible} to $G$, denoted $F \redP G$ if $F$ is representable by a term of $\lambda(\p \cup \{G\})$, and that $F$ is \p-{\em equivalent} to $G$, denoted $F \equivP G$, if $F \redP G$ and $G \redP F$.
\end{definition}
We regularly use that \p-reducibility is a transitive relation, which is easily verified.
We refer to the class of functionals representable by a term from $\lambda(\p\cup\{G\})$ as the class of functionals \emph{generated by} $G$.
Clearly two functionals are \p-equivalent if and only if they generate the same classes of functionals.

\section{The Cook-Urquhart recursor and bounded iteration}\label{Cook Urquhart}

Our starting point is the recursor that Cook and Urquart use to characterize a class of higher-order polynomial-time functionals \cite{CU}.
This recursor is patterned on the scheme of limited recursion on notation introduced by Cobham \cite{Cobham} and its second-order variant, introduced by Mehlhorn \cite{Mehlhorn}. In \cite{KapThes} it is proved that the type-two functionals definable in the Cook-Urquhart system coincide with Mehlhorn's class.
The recursor is defined as follows:
\[
\Rec(\varphi, \psi, \str a, \epsilon) := \str a\quad \text{and}\quad
\Rec(\varphi,\psi,\str a, \str{c}i) := \lmin(\varphi(\str{c}i,\str t),\psi(\str{c}i)),\quad\text{where} \quad \str t = \Rec(\varphi, \psi, \str a, \str c).
\]
Here, the length minimum $\lmin$ returns its left argument if it has strictly smaller length and the right argument otherwise as defined in the preliminaries.
The schemes used by Cobham and Mehlhorn feature explicit external bounding that captures almost directly the notion of bounding by a polynomial (in the first-order setting we could easily use a scheme with explicit bounding by polynomials in the argument size, and as shown by \cite{Ign} this may be extended to the second-order setting as well).
In the Cook-Urquart recursor, this limiting is realized via an additional type-1 input $\psi$.
Our first observation is that the limiting may instead be realized through a type-0 input.

\begin{lemma}[$\Rec \equivP \Rec_0$]\label{bounding with constants}
  The Cook-Urquart recursor and its restriction to constant bounding functions generate the same class of functionals.
  More specifically $\Rec$ is \p-equivalent to the functional
  \[ \Rec_0 (\varphi, \str b, \str a, \str c) := \Rec (\varphi, \lambda \str d. \str b, \str a, \str c). \]
\end{lemma}

\begin{proof}
  From the definition of $\Rec_0$ it is immediate that $\Rec_0 \redP \Rec$.
  To see that also $\Rec \redP \Rec_0$ argue that it suffices to show that $\max \redP \Rec_0$, where $\max$ is the functional that maximizes return values of a function over the initial segments of a string, i.e.\ is recursively defined via $\max(\psi, \epsilon):= \psi(\epsilon)$ and
  \[
  \max(\psi, \str c i) := \begin{cases} \psi (\str c i) & \text{ if } \lmin(\max(\psi, \str c),\psi(\str c i)) = \max(\psi, \str c) \\
    \max(\psi, \str c) &\text{ otherwise.} \end{cases}
  \]
  Indeed, once this is proven $\Rec \redP \Rec_0$ follows from the equality
  \[
  \Rec(\varphi, \psi, \str a, \str c) = \Rec_0(\lambda \str d.\lambda \str t.\lmin(\varphi(\str d, \str t),\psi(\str d)), \mathtt{0}\max(\psi, \str c), \str a, \str c),
  \]
  where the second argument is the maximum with an additional digit added to make sure it is always strictly bigger than any value of $\psi$ on an initial segement of $\str c$.
  This equality can be proven by an induction where the crucial point in the induction step is that the outer of the nested minima always chooses its left argument as value.
  
  To see that the length maximization functional is definable using $\Rec_0$, note that the $\mathsf{argmax}$ functional, which returns the smallest initial segment where a given input-function assumes its maximum, can be defined from $\Rec_0$ via
  \[
  \mathsf{argmax}(\psi,\str c) = \Rec_0(\lambda \str d.\lambda \str t.A(\psi,\str d, \str t), \str c, \epsilon, \str c)
  \]
  where $A(\psi,\str d,\str t) = \str d$ if $\lmin(\psi(\str t),\psi(\str d)) \neq \psi(\str t)$ and $\str t$ otherwise.
  Since $\max(\psi, \str c) = \psi(\mathsf{argmax}(\psi, \str c))$, it follows that $\max$ can be expressed and thus that $\Rec \redP \Rec_0$.
\end{proof}

As a further simplification of $\Rec$, it is possible to eliminate the reference to the current value of the recursion parameter at each step, that is, to replace a functional capturing a form primitive recursion on notation with one that captures functional iteration. 
This is known as a folklore result, but to the best of our knowledge does not appear explicitly in any previous work in this setting.
The most similar characterization we are aware of is one based on typed loop-programs and appeared in \cite{CK}.
In the case of primitive recursion, the equivalence with iteration was first explictly proved in \cite{robinson1947}.

For a function $\varphi\colon \Sigma^*\to\Sigma^*$ let the $n$-fold iteration $\varphi^n$ be inductively defined by $\varphi^0(\str a) := \str a$ and $\varphi^{n+1}(\str a) := \varphi(\varphi^n(\str a))$.
An unbounded iterator would be a functional that takes $n, \str a$ and $\varphi$ as inputs and returns $\varphi^n(\str a)$.
Recall from the introduction, that there are polynomial-time computable $\varphi$ such that the function $\lambda \str a.\lambda \str b. \varphi^{|\str b|}(\str a)$ exhibits exponential growth and is in particular not polynomial time computable.
Thus, to capture the class of feasible functionals, the considered iterator needs to be bounded.
We define the {\em bounded iterator}
$\Iter$ by
\[
\Iter(\varphi,\str b, \str a, \str c):= (\lambda \str t.\lmin(\varphi(\str t),\str b))^{|\str c|}(\lmin(\str a,\str b)).
\]
That is: $\Iter$ performs $|\str c|$ iterations of the input function $\varphi$ on starting value $\lmin(\str a,\str b)$, truncating the resulting value after each iteration to be no longer than the bound $\str b$.

Before we go on to prove the bounded iterator equivalent to the Cook-Urquart recursor, let us briefly discuss the choices we have taken in bounding.
First off, it is easy to see that whether or not the starting value is truncated is irrelevant up to \p-equivalence.
Furthermore, the definition of $\Iter$ is such that the bounding is done after $\varphi$ is applied.
Another possibility would be to consider an iterator where the bounding is done on the argument side of $\varphi$, i.e.\ before its application.
We give a short proof that the resulting iterator is equivalent.
\begin{lemma}($\Iter\equivP\Jter$)
  Output-bounded iteration generates the same class of functionals as argument-bounded iteration.
  More specifically $\Iter$ is \p-equivalent to the functional
  \[
  \Jter(\varphi,\str b, \str a, \str c) := (\lambda \str t.\varphi(\lmin(\str t, \str b)))^{|\str c|}(\str a).
  \]
\end{lemma}
\begin{proof}
  We prove that for all $\varphi, \str b, \str a, \str c$,
  \begin{align*}
    \Iter(\varphi, \str b, \str a, \str c) &= \lmin(\Jter(\varphi, \str b, \str a, \str c),\str b)\tag{*}\\
    \Jter(\varphi, \str b, \str a, \str c) &=\left\{\begin{array}{ll}\str b & \text{if $\str c=\epsilon$;}\tag{**}\\
    \varphi(\Iter(\varphi, \str b, \str a, \str c')) & \text{ if } \str c = \str c' i.
    \end{array}\right.
  \end{align*}
  We prove  (*) and (**) simultaneously by induction on $|\str c|$.
  The case when $|\str c|=0$ is clear, so suppose (*) and (**) hold for all $\str c'$ with $|\str c'|=k \ge 1$.
  Consider $\str c$ with $|\str c|=k+1$, say $\str c = \str{c'}i$ where $|\str c'|=k$.
  Then
  \begin{align*}
    \Iter(\varphi,\str b,\str a, \str c)&=\lmin(\varphi(\Iter(\varphi,\str b,\str a, \str c')), \str b)\\
    &=\lmin(\Jter(\varphi, \str b, \str a, \str c),\str b)\tag{By IH (**)}
  \end{align*}
  and
  \begin{align*}
    \Jter(\varphi, \str b, \str a, \str c)&=\varphi(\lmin(\Jter(\varphi, \str b, \str a, \str c'), \str b))\\
    &=\varphi(\Iter(\varphi,\str b, \str a, \str c')).\tag{By IH (*)}
  \end{align*}
\end{proof}

We end the section with the proof that the bounded iterator, and its modification from the preceding lemma, generate the basic feasible functionals.
That is, that they generate the same class of functionals as the Cook-Urquart recursor.
\begin{lemma}[$\Rec \equivP \Iter$]\label{iterec}
  The Cook-Urquart recursor and the bounded iterator are \p-equivalent.
\end{lemma}
\begin{proof}
  The first implication, namely that $\Iter \redP \Rec$, follows from the equality
  \[ \Iter(\varphi,\str b,\str a, \str c) = \Rec(\lambda \str d.\lambda\str t.\varphi(\str t), \lambda \str d.\str b, \lmin(\str a,\str b), \str c), \]
  that can be proven through a simple induction.

  For the converse note that, by Lemma \ref{bounding with constants} the recursor is equivalent to its version $\Rec_0$ where the bounding is done via a constant instead of a function.
  Thus it suffices to prove that $\Rec_0 \redP \Iter$.
  Furthermore note how close the expanded definition of the iterator is to the definition of $\Rec_0$:
  \[
  \Iter(\varphi, \str b, \str a, \epsilon):=\lmin(\str a,\str b) \quad \text{and} \quad
  \Iter(\varphi, \str b, \str a, \str{c}i):=\lmin(\varphi(\str t),\str b),\quad
  \text{where}\quad t = \Iter(\varphi, \str b, \str a, \str c) \]
  The main difference is that for the recursor the step function $\varphi$ is additionally given access to the value of the recursion parameter $\str c$.
  We postpone the discussion of how to accomodate the additional bounding of the initial value to the end of the proof and show that the operator $\Rec_0'$ defined by
  \[
  \Rec_0'(\varphi,\str a, \str b, \epsilon) := \lmin(\str{a},\str{b})\quad \text{and}\quad
  \Rec_0'(\varphi,\str a, \str b, \str{c}i) := \lmin(\varphi(\str{c}i,\str t),\str b),\quad
  \text{where}\quad \str t = \Rec_0'(\varphi, \str a, \str b, \str c). \]
  can be expressed by using the bounded iterator.
  To achieve this define 
  \[
  \Phi(\varphi, \str b, \str c) := \lambda \langle \str u, \str v \rangle.\langle \str u\mathtt{0}, \lmin(\varphi(\str c^{\leq |\str u|+1},\str v),\str b)\rangle.
  \]
  In the above $\varphi$ has the type of a functional input of the recursor $\Rec_0'$ and $\Phi(\varphi,\str b, \str c)$ has the type of a functional input for the bounded iterator for fixed $\varphi$, $\str b$ and $\str c$.
  We claim that for any $\str c$, and $\str c' \subseteq \str c$,
  \begin{equation}\tag{*}\label{Recursor Iterator}
    \Iter(\Phi(\varphi, \str b, \str c),\langle \mathtt{0}^{|\str c|}, \str b \rangle, \langle \epsilon, \str a\rangle, \str c')
    = \langle \mathtt{0}^{|\str c'|}, \Rec_0'(\varphi, \str b, \str a, \str c')\rangle
  \end{equation}
  and so, in particular
  \[
  \Rec_0'(\varphi, \str b, \str a, \str c) = \pi_2(\Iter(\Phi(\varphi, \str b, \str c), \langle \mathtt{0}^{|\str c|},\str b\rangle, \langle \epsilon, \str a\rangle, \str c)),
  \]
  which proves the \p-reducibility of $\Rec_0'$ to $\Iter$.
  The equality \eqref{Recursor Iterator} can be verified by fixing $\str c$, an proving the following statement by induction on $\str c'$:
  if $\str c' \subseteq \str c$, then \eqref{Recursor Iterator} holds for $\str c'$.
  The base case of this induction follows from the properties we demanded the pairing functions to have.
  Next suppose that the assertion is true for $\str c'$.
  If $\str c'i \subseteq \str c$, then it is also the case that $\str c' \subseteq \str c$, and the induction hypothesis implies that \eqref{Recursor Iterator} holds for $\str c'$.
  But then
  \begin{align*}
    \Iter(\Phi(\varphi, \str b, \str c),\langle \mathtt{0}^{|\str c|}, \str b\rangle, \langle \epsilon, \str a\rangle, \str c'i)&= \lmin(\Phi(\varphi, \str b, \str c)(\langle\mathtt{0}^{|\str c'|}, \Rec_0'(\varphi, \str b, \str a, \str c')\rangle),\langle \mathtt{0}^{|\str c|}, \str b\rangle)\\
    &= \lmin(\langle \mathtt{0}^{|\str{c'}|}\mathtt{0},\lmin(\varphi(\str c^{\leq |\str c'|+1},\Rec_0'(\varphi, \str b, \str a, \str c')), \str b)\rangle,\langle \mathtt{0}^{|\str c|}, \str b\rangle)\\
    &= \lmin(\langle \mathtt{0}^{|\str{c}'i|}, \lmin(\varphi(\str{c'}i, \Rec_0'(\varphi, \str b, \str a, \str c')), \str b)\rangle, \langle \mathtt{0}^{|\str c|}, \str b\rangle)\\
    &= \lmin(\langle \mathtt{0}^{|\str{c}'i|}, \Rec_0'(\varphi, \str b, \str a, \str{c}'i))\rangle, \langle \mathtt{0}^{|\str c|}, \str b\rangle)\\
    &= \langle \mathtt{0}^{|\str{c}'i|}, \Rec_0'(\varphi, \str b, \str a, \str{c}'i))\rangle
  \end{align*}
  Where the last equality uses the properties of the tupling functions again and the fact that $\Rec_0'(\varphi, \str b, \str a, \str {c'}i)$ is either strictly shorter than $\str b$ or equal to $\str b$.

  Finally, to change the initial value, define $H$ as follows:
  \[
  H(\varphi, \str d, \str t, \str a)=\begin{cases}
  \varphi(\str d, \str t) & \text{ if $|\str t| >1$;}\\
  \varphi(\str d, \str a) & \text{ otherwise,}
  \end{cases}.
  \]
  then
  \[
  \Rec_0(\varphi, \str a, \str b, \str c) = \left\{\begin{array}{ll}
  \str a & \text{if $\str c=\epsilon$;}\\
  \Rec_0'(\lambda \str t.\lambda \str d.H(\varphi,\str d,\str t, \str a), \str a, \str b, \str c)  &\text{otherwise}
  \end{array}\right.
  \]
  and thus we obtain that $\Rec_0 \redP \Rec_0'\redP \Iter$ and with the fact $\Rec\redP \Rec_0$ from Lemma \ref{bounding with constants} also the desired \p-reducibility $\Rec \redP \Iter$.
\end{proof}

\section{Iteration with Constant Length Revision}\label{length revision}

Both the Cook-Urquart recursor $\Rec$ as well as the bounded iterator $\Iter$ require an absolute bound on the size of intermediate value encountered during a recursion.
Specifying such a bound {\em a priori} can be cumbersome and this section provides an alternative way of bounding an iteration that is inspired by the classes \spt and \mpt we considered in earlier work \cite{DBLP:conf/rta/KawamuraS17, KS}.
The elementary notion used in the definition of \spt is that of a length revision.
In an \otm computation a length revision is encountered whenever the answer to an oracle query is longer than any previous response.
This notion of a length revision can easily be translated to the realm of recursion schemes:
in a recursive definition a length revision happens when the return value of the step function is bigger than any of the values returned earlier.
%\OMIT{
%First, we give a precise translation of the notion of length revisions to the iterative setting.
%Fix some function $\varphi\colon\Sigma^* \to \Sigma^*$.
%We say that a \emph{length revision} is encountered at the $n$-th iteration of $\varphi$ on $\str a$ if $|\varphi^n(\str a)|$ is strictly bigger than any $|\varphi^i(\str a)|$ with $i < n$, that is, the size of the result of the call to $\varphi$ exceeds $|\str a|$ and the size of any previous call.
%}
%Recall the we have defined simple function iteration $\varphi^n(\str a)=\underbrace{\varphi(\varphi(\dots\varphi}_{\text{$n$ times}}(\str a)\dots))$.
 %In this setting, a length revision is incurred when the length of an answer from $\varphi$ exceeds the maximum such obtained so far. 
In particular,  define
\[
\varphi^n_{!k}(\str a) := \underbrace{\varphi(\varphi(\dots\varphi}_{\text{$\ell$ times}}(\str a)\dots))
\]
where $\ell\le n$ is maximum such that the sequence of applications contains no more than $k$ {\em length revisions}, that is an application $\varphi(\str t)$ where $|\varphi(\str t)|$ exceeds $|\str a|$ and $|\varphi(\str t')|$ for any previous call. 
 In particular, when $k=0$, this means that no calls return a value that exceeds $|\str a|$.
For $k \ge 0$, the \emph{$k$-revision iterator $\Iter_k$} is the functional defined by
\[
\Iter_k(\varphi, \str a, \str c) := \varphi_{!k}^{|\str c|}(\str a)
\]
%Note that the restriction that $k \ge 1$ formalizes the convention that the $0$th iteration is always considered a length revision.
%Thus, shifting ensures that the first iterator has a computational value and additionally keeps the meaning of $k$ tied to the number of length revisions:
%the first proper length revision happens when the iteration value outgrows the initial value for the frist time and the $0$-revision iterator stops iterating just before the first length-revision happens.

Superficially, this definition is similar to that of the bounded iterator $\Iter$ from the last section.
The functional $\Iter_k$ iterates a function where the iteration is bounded by $k$ just like the bounded iterator does for each fixed bounding argument $\str b$. 
%Indeed, it seems strange that we chose to treat $k$ as an integer index instead of the length of an additional input.
%On closer inspection it becomes clear that the way that the bounding is done is very different and indeed our goal is to prove that for each fixed $k$ the operator $\Iter_k$ is \p-equivalent to the bounded iterator $\Iter$ (Theorem~\ref{main} below).
The essential difference is that $k$ is a statically fixed parameter, i.e.,
$\Iter_k$ constitutes a family of iterators. Our goal is to show  that for each fixed $k$ the operator $\Iter_k$ is \p-equivalent to the bounded iterator $\Iter$ (Theorem~\ref{main} below). 

% that are of pairwise equal expressive power and only differ in the bounding condition.
%Moreover, this bounding condition becomes less and less restrictive with growing $k$.
%
Without restrictions on $k$, neither of the reducibilities required to prove that  $\Iter_k\equivP \Iter$ are obvious.
However, the claim that $\Iter_k\redP \Iter$ should appear reasonable given the \otm-based characterization of the basic feasible functionals \cite{KC}.
As proven in the last section, $\Iter$ is \p-equivalent to $\Rec$ and thus it is enough to check that $\Iter_k$ is a basic feasible functional, i.e., that $\Iter_k$ is computable by an \otm whose run-time is bounded by a second-order polynomial. This may be done in a straightforward way, but it is important to note that the complexity of the bounding polynomial (in terms of the depth of calls to the function input, rather than the degree) increases with $k$. In particular, while $\Iter_k$ is \p-equivalent to $\Iter$ for every $k$, the revision parameter $k$ provides a finer delineation of expressive power.
%This partly explains why the direct scheme-based proof that $\Iter_k \redP \Iter$ given below requires some work and also why the case $k=0$ is special and can be taken care of right away:

Without an appeal to the \otm-based characterization, showing the \p-equivalence of $\Iter_k$ and $\Iter$ becomes more of a challenge, although the case for $\Iter_0$ is relatively straightforward:
\begin{lemma}[$\Iter_0\redP \Iter$]\label{induction start}
  The $0$-revision iterator is \p-reducible to the bounded iterator.
\end{lemma}
\begin{proof}
  The main hurdle is to account for the difference in how the violation of the bound is realized: $\Iter_0$ defaults to the previous value in the iteration while $\Iter$ defaults to the value it is given as bound.
  Set
  \[
  G(\varphi,\str t\mathtt{1},\str b) := \str t\mathtt{1}\quad\quad\quad\quad
  G(\varphi,\str t\mathtt{0},\str b) :=\left\{\begin{array}{ll}
  \varphi(\str t)\mathtt{0} & \text{if $|\varphi(\str t)| \le |\str b|$;}\\
  \str{t}\mathtt{1} &\text{otherwise.}
  \end{array}\right.
  \]
  Then $\Iter_0(\varphi, \str a, \str c)$ can be obtained from $\Iter(\lambda \str t.G(\varphi,\str t, \str a), \str a, \str a\mathtt{0}, \str c)$ by simply dropping the last bit.
  Since the definition of $G$ only uses type-1 polynomial time operations and application, the \p-reducibility follows.
\end{proof}
%A very similar argument shows that $\Iter_0$ can be defined from any other $k$-revision iterator.
%\begin{cor}
%  For all $k \ge 1$, $\Iter_0 \redP \Iter_k$.
%\end{cor}
%As for the bounded iterator, the value of the $k$-revision iterators depends only on the length of the last argument:
%$\Iter_k(\varphi, \str a, \str c) = \Iter_k(\varphi, \str a, \sdzero^{|\str c|})$, this makes the following shorthand usefull:
%\[ \varphi_k^n(\str a) := \Iter_k(\varphi, \str a, \sdzero^n) \]
Note that for unrestricted iteration it holds that $\varphi^n(\varphi^{n'}(\str a)) = \varphi^{n + n'}(\str a)$.
The following observation points out a similar additivity property for $\varphi^n_{!k}$ and is the starting point for recursively constructing \p-reductions of $\Iter_k$ to $\Iter$:
\begin{lemma}
\label{unwind}
For given $\varphi$, $\str a$ and numbers $k$ and $n$ set
\[ \ell := \min\{i\mid \forall j, i \leq j \leq n \Rightarrow \varphi^i_{!k}(\str a)=\varphi^j_{!k}(\str a)\}, \]
then $\ell \leq n$ and it holds that
\[
\varphi^n_{!(k+1)}(\str a)=\varphi^{n-\ell}_{!1}(\varphi^{\ell}_{!k}(\str a)).\tag{$*$}
\]
\end{lemma}
\begin{proof}
  Since $n$ always fulfills the condition in the minimization, it follows that $\ell \leq n$.
  To prove $(*)$ we first note that the minimization condition may be satisfied in two different ways.
  It may be the case that $\varphi$ gives the same return value on all of the strings $\varphi_{!k}^\ell(\str a),\dots, \varphi_{!k}^{n-1}(\str a)$.
  In this case there will be no further length revisions, and so $\varphi^{n-\ell}_{!1}(\varphi^{\ell}_{!k}(\str a))=\varphi^{n-\ell}(\varphi^{\ell}_{!k}(\str a))=\varphi^n_{!(k+1)}(\str a)$.
  Thus suppose that it is not the case that $\varphi$ is constant on these strings and let $j$ be such that $\ell \leq j \leq n-1$ and $\varphi(\varphi^j_{!k}(\str a)) \neq \varphi^{\ell+1}_{!k}(\str a)$.
  By definition of $\ell$, the strings $\varphi_{!k}^\ell(\str a),\ldots,\varphi_{!k}^n(\str a)$ are still all equal.
  Then $\varphi(\varphi^j_{!k}(\str a)) = \varphi(\varphi^l_{!k}(\str a))$ and $\varphi(\varphi^l_{!k}(\str a))$ and $\varphi^{l+1}_{!k}(\str a)$ can only be different if the $(\ell+1)$-st call to $\varphi$ triggers the $(k+1)$-st length revision.
  Thus, in this case $\varphi_{!1}^{m}(\varphi^\ell_k(\str a)) = \varphi(\varphi^\ell_{!k}(\str a)) = \varphi^{\ell+m}_{!(k+1)}(\str a)$ for any $m$ and in particular $(*)$ must hold.
\end{proof}

\noindent
In fact, the above proof proves the following slightly stronger statement.

\medskip

\begin{cor}
  \label{unwind2}
  The equality in the last Lemma may be replaced by $\varphi^n_{!(k+1)}(\str a)=\varphi^{n-\ell-1}_{!0}(\varphi(\varphi^{\ell}_{!k}(\str a)))$.
\end{cor}

\medskip

\noindent
This allows us to establish the following.
\begin{lemma}[$\Iter_k \redP \Iter$] For $k \ge 1$, the $k$-revision iterator is \p-reducible to the bounded iterator.
\label{iteriter}
\end{lemma}
\begin{proof}
  We proceed by induction on $k$.
  The case of $k=0$ has been taken care of in Lemma \ref{induction start}.
  Suppose that the Lemma holds for $k$.
  We must now define $\Iter_{k+1}$ using $\Iter$.
  By Lemma~\ref{unwind} it is sufficient to show that there exists a function that on inputs $\varphi$, $\str a$, $k$ and $n$ returns the value $\ell$ and is \p-reducible to $\Iter$.
  First note that the condition $\varphi^i_k(\str a)=\varphi^j_k(\str a)$ can be checked by a function from $\lambda(\p \cup \{\Iter_k\}) \subseteq \lambda(\p \cup \{\Iter\})$, where the inclusion follows by the induction hypothesis.
  Now all that remains is to use $\Iter$ to characterize the bounded quantification and search used to define $\ell$ in Corollary~\ref{unwind2}.
  Define the following functionals:
  \begin{align*}
    U(\psi, \str a, \str c)&:=\left\{\begin{array}{ll}
    \epsilon & \text{if $\forall_{i \le |\str c|}\left(\psi(|\mathtt{0}^i,\str a)|=0\right)$;}\\
    \mathtt{0} & \text{otherwise.}
    \end{array}\right.\\
    M(\psi, \str a, \str c)&:= \mathtt{0}^j\ \text{where $j=\mu_{i \le |\str c|}\left(|\psi(\mathtt{1}^i,\str a)| > 0\right)$ if such $i$ exists, and $i+1$ otherwise.}
  \end{align*}
  We first show that $U\redP \Rec$ and appeal to Lemma~\ref{iterec} to see that it is \p-reducible to $\Iter$. 
  Define 
  \[
  V(\psi,  \str a, \str t, \str d):=\left\{\begin{array}{ll}
  \epsilon & \text{if $\str t = \epsilon$ and $|\psi(\mathtt{0}^{|\str d|},\str a)|=0$;}\\
  \mathtt{0} & \text{otherwise.}
  \end{array}\right.
  \]
  Then $U(\psi, \str a, \str c)=\Rec(\lambda \str t.\lambda \str d.V(\psi, \str a, \str t, \str d), \lambda \str d.\mathtt{0}, \epsilon, \str c)$.
  Since the definition of $V$ only uses polynomial-time computable type-1 functions and application we conclude $U \redP \Rec$.
  To show that $M \redP \Iter$, first define
  \[
  N(\psi, \str a, \str t\mathtt{0}):=\str t\mathtt{0}\quad\quad\quad\quad
  N(\psi, \str a, \str t\mathtt{1}):=\left\{\begin{array}{ll}
  \mathtt{0}\str t & \text{if $|\psi(\mathtt{0}^{|\str t|\cut 1},\str a)| \le 0$;}\\
  \str t \gg 1 &\text{otherwise,}
  \end{array}\right.
  \]
  and define
  \[
  A(\psi, \str a):=\left\{\begin{array}{ll}
  \mathtt{0} & \text{if $|\psi(\epsilon,\str a)| > 0$;}\\
  \mathtt{01} & \text{otherwise.}
  \end{array}\right.
  \]
  Then $M(\psi, \str a, \str c) = \str m =\Iter(\lambda \str t.N(\psi, \str a, \str t), A(\psi, \str a), \str{c}\mathtt{0}, \str{c})\gg 1$. In particular, if $\str m$ ends in $\mathtt{0}$, then $\str m=\mathtt{0}^{j+1}$ and if it ends in $\mathtt{1}$ then $\str m=\mathtt{0}^{|\str c|+1}\mathtt{1}$.
\end{proof}

\section{Iteration with Constant Lookahead Revision}\label{lookahead revision}
Moving to lookahead revision, the definition is similar. 
Consider the following variant of function iteration
\[
\varphi^n_{?k}(\str a) := \underbrace{\varphi(\varphi(\dots\varphi}_{\text{$\ell$ times}}(\str a)\dots))
\]
where $\ell\le n$ is maximum such that the sequence of applications contains no more than $k$ {\em lookahead revisions}, that is an application $\varphi(\str t)$ where $|\str t|$ exceeds $|\str t'|$ for any previous call $|\varphi(\str t')|$. Note that we have not included the initial call $\varphi(\str a)$ as a lookahead revision (choosing to do so would not change any results below.) Then, for $k \ge 0$, the \emph{$k$-lookahead revision iterator $\Jter_k$} is the functional such that
\[
\Jter_k(\varphi, \str a, \str c) = \varphi_{?k}^{|\str c|}(\str a)
\]
We now consider the relative power of $\Jter_k$.
\begin{lemma} For any $k \ge 0$,
$\Jter_k \redP \Iter_k$.
\end{lemma}
\begin{proof}
We claim that $\Jter_k(\varphi, \str a, \str c)=\varphi(\Iter_k(\varphi, \str a,  \str c\gg 1))$. This is clear in the case that there are no more than $k$ length revisions in the evaluation of $\Iter_k(\varphi, \str a, \str b, \str c\gg 1)$, as any lookahead revision corresponds exactly to a preceding length revision,  and so
$\Jter_k(\varphi, \str a, \str c)= \varphi^{|\str c|}(\str a) =\varphi(\Iter_k(\varphi, \str a,  \str c\gg 1))$. Otherwise suppose that $\Iter_k(\varphi, \str a, \str b, \str c\gg 1)=\varphi^\ell(\str a)$, which means in particular that  $\ell$ is the minimum value less than $|\str c|$ such that evaluating $\varphi^{\ell+1}(\str a)$ results in $k+1$ length revisions. But then $\ell$ is the minimum value less than $|\str c|$ such that  $\varphi^{\ell+2}(\str a)$ results in $k+1$ length revisions. But this means 
$\Jter_k(\varphi, \str a, \str c)= \varphi^{\ell+1}(\str a) =\varphi(\Iter_k(\varphi, \str a,  \str c\gg 1))$.
\end{proof}
\begin{lemma} 
For any $k \ge 0$, $\Jter \redP \Jter_k$.
\end{lemma}
\begin{proof}
Unfortunately, the situation is a little less straightforward than we might hope, as $\Jter$ and $\Jter_k$ differ slightly in the way they do bounding. $\Jter_k$ expects queries to be bounded in length by previous queries while $\Jter$ uses an explicit bound $\str b$. Define $\psi$ as follows:
\[
\psi(\str t\mathtt{0},\str a, \str b) := \lmin(\str a, \str b)\mathtt{1}\quad\quad\quad\quad\psi(\str t\mathtt{1},\str a, \str b) := \lmin(\varphi(\str t), \str b)\mathtt{1}
\]

\noindent
{\it Claim}. For all $\str a, \str b, \str c$, $\Jter_k(\lambda \str t.\psi(\str t, \str a, \str b), \str b\mathtt{0}, \mathtt{0}\str c)=\lmin(\Jter(\varphi,\str a, \str b, \str c), \str b)\mathtt{1}$

\medskip

\noindent
To prove this claim, first note the in the iteration on the left, the first call to $\psi$ is $\str b\mathtt{0}$, of length $|\str b|+1$. All subsequent calls are clearly bounded by $|\str b|+1$. So there will be no lookahead revisions in this iteration and it remains to prove equality without consideration of the lookahead bound $k$. We use induction on $\str c$. When $\str c = \epsilon$,
\begin{align*}
\Jter_k(\lambda \str t.\psi(\str t, \str a, \str b), \str b\mathtt{0}, \mathtt{0}\str c) &=\Jter_k(\lambda \str t.\psi(\str t, \str a, \str b), \str b\mathtt{0}, \mathtt{0})\\
&=\psi(\Jter_k(\lambda \str t.\psi(\str t, \str a, \str b), \str b\mathtt{0}, \epsilon),\str a, \str b)\\
&=\psi(\str b0, \str a, \str b)\\
&=\lmin(\str a, \str b)1\\
&=\lmin(\Jter(\varphi,\str a, \str b, \str c), \str b)\mathtt{1}.
\end{align*}
Now assume that the claim holds for $\str c$. Then
\begin{align*}
\Jter_k(\lambda \str t.\psi(\str t, \str a, \str b), \str b\mathtt{0}, \mathtt{0}\str ci)
&=\psi(\Jter_k(\lambda \str t.\psi(\str t, \str a, \str b), \str b\mathtt{0}, \mathtt{0}\str c),\str a, \str b)\\
&=\psi(\lmin(\Jter(\varphi,\str a, \str b, \str c), \str b)\mathtt{1}, \str a, \str b)\\
&=\lmin(\varphi(\lmin(\Jter(\varphi,\str a, \str b, \str c), \str b)), \str b)\mathtt{1}\\
&=\lmin(\Jter(\varphi,\str a, \str b, \str ci), \str b)\mathtt{1}
\end{align*}
Now define $\Iter''$ as follows:
\[
\Iter''(\varphi,\str a, \str b, \str c)=\left\{\begin{array}{ll}
			\str a & \text{if $\str c = \epsilon$;}\\
			\varphi(\Jter_k(\lambda \str t.\psi(\str t, \str a, \str b), \str b\mathtt{0}, \mathtt{0}(\str c \gg 1))\gg 1) &\text{otherwise.}
\end{array}\right.
\]
First note that $\Iter'' \preceq_{\lambda P} \Jter_k$, as definition by cases is a poly-time operation. When $\str c=\epsilon$, $\Iter''(\varphi,\str a, \str b, \str c)=\str a=\Jter(\varphi,\str a, \str b, \str c)$. Otherwise, by the claim, 
\[
\Jter_k(\lambda \str t.\psi(\str t, \str a, \str b), \str b\mathtt{0}, \mathtt{0}(\str c \gg 1))\gg 1=\lmin(\Jter(\varphi,\str a, \str b, \str c\gg 1), \str b),
\]
so that
\[
\Iter''(\varphi,\str a, \str b, \str c)=\varphi(\lmin(\Jter(\varphi,\str a, \str b, \str c\gg 1), \str b))=\Jter(\varphi,\str a, \str b, \str c).
\]
%Then
%$\Iter(\varphi,\str a, \str b, \str c) =\Iter_0(\lambda\str t.\psi(\str t, \str a, \str b), \str b\mathtt{0},\str c)\gg 1$.
%The low-order bit of the first input to $\psi$ is $\mathtt{0}$ only on the first application. Since the iteration of $\psi$ starts with $\str b\mathtt{0}$, this application of $\psi$ will return $\varphi(\lmin(\str a, \str b))\mathtt{1}$, which is exactly $\Jter(\varphi,\str a, \str b, \epsilon)\mathtt{1}$. In all the remaining iterations of $\psi$, the low-order bit is $\mathtt{1}$, and again the correct value for the bounded iteration of $\varphi$ is returned. The low-order bit is dropped at the end. All queries are bounded in length by $|\str b|$, so there are no lookahead revisions. 
\end{proof}

\medskip

\noindent
Putting everything together, we have a characterization which is the main result of this paper.
\begin{theorem}\label{main} For every $k \ge 0$, $\Iter \equivP \Iter_k \equivP \Jter_k \equivP \Jter$.

\end{theorem}

\section{More efficient approaches}\label{eff}
The implementation of $\Iter_k$ by $\Iter$ given in Lemma~\ref{iteriter} requires considerable overhead, involving a bounded quanitification and bounded search at each step.
An implementation which directly follows this definition is poly-time, but is needlessly complex.
The following observation (which in this setting correpsonds to tail-recursion elimination) will simplify things considerably.
In particular, we note an alternate characterization of $\varphi^n$: $\varphi^0(\str a)=\str a$ and $\varphi^{n+1}(\str a)=\varphi^n(\varphi(\str a))$.
This leads to the following characterization of $\varphi^n_{!k}$.
%The following lemma provides a scheme based on tail-recursion elimination that can be used to generate the values of the $\varphi^n_k$, and thus also of the functionals $\Iter_k$, way more efficiently.

\begin{lemma} For all $n,k \ge 0$ we have
\begin{align*}
\varphi^0_{!k}(\str a) &= \str a\\
\varphi^{n+1}_{!0}(\str a)&=\left\{\begin{array}{ll}
	\str a & \text{if $|\varphi(\str a)|>|\str a|$;}\\
	\varphi^n_{!0}(\varphi(\str a)) & \text{otherwise.}
\end{array}\right.\\
\varphi^{n+1}_{!(k+1)}(\str a)&=\left\{\begin{array}{ll}
	\varphi^n_{!k}(\varphi(\str a)) & \text{if $|\varphi(\str a)|>|\str a|$;}\\
	\varphi^n_{!(k+1)}(\varphi(\str a)) & \text{otherwise.}
\end{array}\right.
\end{align*}
\end{lemma}

\begin{proof}
  We prove by induction on $k$ that the claim holds for all $n$.
  When $k=0$, iteration stops (absolutely) if $|\varphi(\str a)|>|\str a|$, otherwise it proceeds to the next step.
  Now assume for $k$ that the claim holds for all $n$.
  We show that for $k+1$ it holds for all $n$, by induction on $n$.
  When $n=0$ this is immediate.
  Assume that it holds for $n$, and consider $\varphi^{n+1}_{!(k+1)}(\str a)$.
  Clearly, if $|\varphi(\str a)|\le|\str a|$, no length revision occurs on the first call, and so $k+1$ are still available for the remaining $n$ calls.
  Otherwise, only $k$ length revisions are available for the remaining calls.
\end{proof}
We also note that, implicit in the proof of Lemma~\ref{iteriter}, is an implementation which is also efficient -- in particular, if we ``unwind'' the induction, we are eventually left relying only on $\Iter_0$.
As described in \cite{KS}, \S 4.3, we can implement the resulting definition using a form of ``re-entrant'' recursion.
We may view the violation of the length-revision bound as triggering an exception, which may then be caught by an exception handler which re-starts the recursion at the point after the offending oracle call has taken place.

\section{Conclusions and Future Work}
We have provided a new linguistic characterization of the higher-order polynomial time via iteration schemes that restrict the number of times a step function, presented as an oracle, may return an answer or be presented an input which in length exceeds all previous answers (resp. queries).
The characterization and the methods used to prove it lead to a number of questions and potential directions for future research.

The characterization provided in this paper could be termed {\em intrinsic}, in that no external bounding is present in the iteration schemes $\Iter_k$ and $\Jter_k$.
The condition itself, however, appears to depend on the dynamics of a particular computation. On its face it is not a structural/syntactic restriction, as is usual in implicit computational complexity. This suggests two directions for further research.
The first is to investigate the possibility of statically deriving bounds on query revision.
The second is to investigate distinctions on how computational resources are bounded as suggested by this and related work, for example intrinsic versus extrinsic, dynamic versus static, and feasibly constructive versus non-feasibly constructive (an example of non-feasibly constructive bounding would be the second-order polynomials of \cite{KC}).
A related observation is that iteration with bounded query revision appears to be a generalization of non-size-increasing computation \cite{Hof}.
This apparent connection merits further investigation.

In \S\ref{eff} above, we begin to explore the interplay between familiar programming techiques from the implemenation of functional programming languages (e.g, tail-recursion elimination) with respect to the efficient implementation of our iteration schemes.
We have also noted that the introduction on control primitives (e.g., {\tt catch} and {\tt throw}) may be relevant to the characterization of complexity classes in this setting.
We note that such control operators have been shown in \cite{CCF94} to be relevant to the general characterization of {\em sequential} higher-order computation. Here we only scratch the surface.
Further investigation of these and related techniques in the context of linguistic characterizations of computational complexity could prove fruitful.

As noted at several points in our development, there are issues of finer-grained complexity that arise from our translations. This gives rise to natural questions on the efficiency, or syntactic complexity, of translations, which bear further investigation. 

Finally, while we have drawn an analogy between \otms with bounded query revision (as introduced in \cite{KS}) and certain recursion schemes, we have not investigated just how closely related they are.
While the equivalences proved in \cite{KS} and in this paper imply an equivalence for all the models, a direct proof would be very interesting in furthering our understanding of poly-time \otms.
It would be very rewarding if a simplified proof of the equivalence of \cite{KC} could be obtained in this setting.

%\bibliographystyle{eptcs}
%\bibliography{DICE}

\begin{thebibliography}{10}
\providecommand{\bibitemdeclare}[2]{}
\providecommand{\surnamestart}{}
\providecommand{\surnameend}{}
\providecommand{\urlprefix}{Available at }
\providecommand{\url}[1]{\texttt{#1}}
\providecommand{\href}[2]{\texttt{#2}}
\providecommand{\urlalt}[2]{\href{#1}{#2}}
\providecommand{\doi}[1]{doi:\urlalt{http://dx.doi.org/#1}{#1}}
\providecommand{\bibinfo}[2]{#2}

\bibitemdeclare{article}{BC}
\bibitem{BC}
\bibinfo{author}{Stephen \surnamestart Bellantoni\surnameend} \&
  \bibinfo{author}{Stephen~A. \surnamestart Cook\surnameend}
  (\bibinfo{year}{1992}): \emph{\bibinfo{title}{A New Recursion-Theoretic
  Characterization of the Polytime Functions}}.
\newblock {\sl \bibinfo{journal}{Computational Complexity}}
  \bibinfo{volume}{2}, pp. \bibinfo{pages}{97--110}, \doi{10.1007/BF01201998}.

\bibitemdeclare{article}{CCF94}
\bibitem{CCF94}
\bibinfo{author}{R.~\surnamestart Cartwright\surnameend}, \bibinfo{author}{P.L.
  \surnamestart Curien\surnameend} \& \bibinfo{author}{M.~\surnamestart
  Felleisen\surnameend} (\bibinfo{year}{1994}): \emph{\bibinfo{title}{Fully
  Abstract Semantics for Observably Sequential Languages}}.
\newblock {\sl \bibinfo{journal}{Information and Computation}}
  \bibinfo{volume}{111}(\bibinfo{number}{2}), pp. \bibinfo{pages}{297 -- 401},
  \doi{10.1006/inco.1994.1047}.

\bibitemdeclare{incollection}{Cobham}
\bibitem{Cobham}
\bibinfo{author}{A.~\surnamestart Cobham\surnameend} (\bibinfo{year}{1965}):
  \emph{\bibinfo{title}{The intrinsic computational difficulty of functions}}.
\newblock In \bibinfo{editor}{Yehoshua \surnamestart Bar{-}Hillel\surnameend},
  editor: {\sl \bibinfo{booktitle}{Logic, Methodology and Philosophy of
  Science: Proc. 1964 Intl. Congress (Studies in Logic and the Foundations of
  Mathematics)}}, \bibinfo{publisher}{North-Holland Publishing}, pp.
  \bibinfo{pages}{24--30}.

\bibitemdeclare{incollection}{Cook}
\bibitem{Cook}
\bibinfo{author}{S.A. \surnamestart Cook\surnameend} (\bibinfo{year}{1992}):
  \emph{\bibinfo{title}{Computability and complexity of higher type
  functions}}.
\newblock In: {\sl \bibinfo{booktitle}{Logic from computer science ({B}erkeley,
  {CA}, 1989)}}, {\sl \bibinfo{series}{Math. Sci. Res. Inst.
  Publ.}}~\bibinfo{volume}{21}, \bibinfo{publisher}{Springer, New York}, pp.
  \bibinfo{pages}{51--72}, \doi{10.1007/978-1-4612-2822-6\_3}.

\bibitemdeclare{incollection}{CK}
\bibitem{CK}
\bibinfo{author}{S.A. \surnamestart Cook\surnameend} \& \bibinfo{author}{B.M.
  \surnamestart Kapron\surnameend} (\bibinfo{year}{1990}):
  \emph{\bibinfo{title}{Characterizations of the basic feasible functionals of
  finite type}}.
\newblock In: {\sl \bibinfo{booktitle}{Feasible mathematics ({I}thaca, {NY},
  1989)}}, {\sl \bibinfo{series}{Progr. Comput. Sci. Appl.
  Logic}}~\bibinfo{volume}{9}, \bibinfo{publisher}{Birkh\"auser}, pp.
  \bibinfo{pages}{71--96}, \doi{10.1007/978-1-4612-3466-1\_5}.

\bibitemdeclare{article}{CU}
\bibitem{CU}
\bibinfo{author}{S.A. \surnamestart Cook\surnameend} \&
  \bibinfo{author}{A.~\surnamestart Urquhart\surnameend}
  (\bibinfo{year}{1993}): \emph{\bibinfo{title}{Functional interpretations of
  feasibly constructive arithmetic}}.
\newblock {\sl \bibinfo{journal}{Ann. Pure Appl. Logic}}
  \bibinfo{volume}{63}(\bibinfo{number}{2}), pp. \bibinfo{pages}{103--200},
  \doi{10.1016/0168-0072(93)90044-E}.

\bibitemdeclare{article}{Hof}
\bibitem{Hof}
\bibinfo{author}{Martin \surnamestart Hofmann\surnameend}
  (\bibinfo{year}{2003}): \emph{\bibinfo{title}{Linear types and
  non-size-increasing polynomial time computation}}.
\newblock {\sl \bibinfo{journal}{Inf. Comput.}}
  \bibinfo{volume}{183}(\bibinfo{number}{1}), pp. \bibinfo{pages}{57--85},
  \doi{10.1016/S0890-5401(03)00009-9}.

\bibitemdeclare{article}{Ign}
\bibitem{Ign}
\bibinfo{author}{A.~\surnamestart Ignjatovic\surnameend} \&
  \bibinfo{author}{A.~\surnamestart Sharma\surnameend} (\bibinfo{year}{2004}):
  \emph{\bibinfo{title}{Some applications of logic to feasibility in higher
  types}}.
\newblock {\sl \bibinfo{journal}{ACM TOCL}}
  \bibinfo{volume}{5}(\bibinfo{number}{2}), pp. \bibinfo{pages}{332--350},
  \doi{10.1145/976706.976713}.

\bibitemdeclare{article}{KC}
\bibitem{KC}
\bibinfo{author}{B.M. \surnamestart Kapron\surnameend} \& \bibinfo{author}{S.A.
  \surnamestart Cook\surnameend} (\bibinfo{year}{1996}):
  \emph{\bibinfo{title}{A new characterization of type-{$2$} feasibility}}.
\newblock {\sl \bibinfo{journal}{SIAM J. Comput.}}
  \bibinfo{volume}{25}(\bibinfo{number}{1}), pp. \bibinfo{pages}{117--132},
  \doi{10.1137/S0097539794263452}.

\bibitemdeclare{incollection}{KS}
\bibitem{KS}
\bibinfo{author}{B.M. \surnamestart Kapron\surnameend} \&
  \bibinfo{author}{F.~\surnamestart Steinberg\surnameend}
  (\bibinfo{year}{2018}): \emph{\bibinfo{title}{Type-two polynomial-time and
  restricted lookahead}}.
\newblock In: {\sl \bibinfo{booktitle}{{P}roceedings of the 33rd {A}nnual
  {ACM/IEEE} {S}ymposium on {L}ogic in {C}omputer {S}cience ({O}xford, {UK}),
  2018}}, \bibinfo{publisher}{{ACM}, New York}, pp. \bibinfo{pages}{579--598},
  \doi{10.1145/3209108.3209124}.

\bibitemdeclare{techreport}{KapThes}
\bibitem{KapThes}
\bibinfo{author}{Bruce~M. \surnamestart Kapron\surnameend}
  (\bibinfo{year}{1991}): \emph{\bibinfo{title}{Feasible Computation in Higher
  Types}}.
\newblock \bibinfo{type}{Technical Report} \bibinfo{number}{249/91},
  \bibinfo{institution}{Computer Science Department, University of Toronto}.

\bibitemdeclare{inproceedings}{DBLP:conf/rta/KawamuraS17}
\bibitem{DBLP:conf/rta/KawamuraS17}
\bibinfo{author}{Akitoshi \surnamestart Kawamura\surnameend} \&
  \bibinfo{author}{Florian \surnamestart Steinberg\surnameend}
  (\bibinfo{year}{2017}): \emph{\bibinfo{title}{Polynomial Running Times for
  Polynomial-Time Oracle Machines}}.
\newblock In: {\sl \bibinfo{booktitle}{2nd International Conference on Formal
  Structures for Computation and Deduction, {FSCD} 2017, September 3-9, 2017,
  Oxford, {UK}}}, pp. \bibinfo{pages}{23:1--23:18},
  \doi{10.4230/LIPIcs.FSCD.2017.23}.

\bibitemdeclare{inproceedings}{Leivant}
\bibitem{Leivant}
\bibinfo{author}{Daniel \surnamestart Leivant\surnameend}
  (\bibinfo{year}{1991}): \emph{\bibinfo{title}{A Foundational Delineation of
  Computational Feasiblity}}.
\newblock In: {\sl \bibinfo{booktitle}{Proceedings of the Sixth Annual {IEEE}
  Symposium on Logic in Computer Science (Amsterdam, The Netherlands), 1991}},
  \bibinfo{publisher}{{IEEE} Computer Society}, pp. \bibinfo{pages}{2--11},
  \doi{10.1109/LICS.1991.151625}.

\bibitemdeclare{article}{Mehlhorn}
\bibitem{Mehlhorn}
\bibinfo{author}{K.~\surnamestart Mehlhorn\surnameend} (\bibinfo{year}{1976}):
  \emph{\bibinfo{title}{Polynomial and abstract subrecursive classes}}.
\newblock {\sl \bibinfo{journal}{J. Comp. Sys. Sci.}}
  \bibinfo{volume}{12}(\bibinfo{number}{2}), pp. \bibinfo{pages}{147--178},
  \doi{10.1016/S0022-0000(76)80035-9}.

\bibitemdeclare{article}{robinson1947}
\bibitem{robinson1947}
\bibinfo{author}{Raphael~M. \surnamestart Robinson\surnameend}
  (\bibinfo{year}{1947}): \emph{\bibinfo{title}{Primitive recursive
  functions}}.
\newblock {\sl \bibinfo{journal}{Bull. Amer. Math. Soc.}}
  \bibinfo{volume}{53}(\bibinfo{number}{10}), pp. \bibinfo{pages}{925--942},
  \doi{10.1090/S0002-9904-1947-08911-4}.

\bibitemdeclare{incollection}{MR1238294}
\bibitem{MR1238294}
\bibinfo{author}{Anil \surnamestart Seth\surnameend} (\bibinfo{year}{1993}):
  \emph{\bibinfo{title}{Some desirable conditions for feasible functionals of
  type {$2$}}}.
\newblock In: {\sl \bibinfo{booktitle}{Eighth {A}nnual {IEEE} {S}ymposium on
  {L}ogic in {C}omputer {S}cience ({M}ontreal, {PQ}, 1993)}},
  \bibinfo{publisher}{IEEE Comput. Soc. Press, Los Alamitos, CA}, pp.
  \bibinfo{pages}{320--331}, \doi{10.1109/LICS.1993.287576}.

\end{thebibliography}

\end{document}